\documentclass[twocolumn,aps,pra,showpacs,superscriptaddress,nofootinbib]{revtex4-1}

\usepackage[latin1]{inputenc}
\usepackage{graphicx,epstopdf}
\usepackage{amsmath}
\usepackage{hyperref}
\usepackage{amsfonts}
\usepackage{amsthm}
\usepackage{amssymb}
\usepackage{color}
\usepackage{latexsym}
\usepackage{times,txfonts}

\newtheorem{theorem}{Theorem}

\newcommand{\fig}[1]{{Fig.}}

\begin{document}

\title{Non-Markovianity through Multipartite Correlation Measures}

\author{Fagner M. Paula}
\email{fagnerm@utfpr.edu.br }
\affiliation{Universidade Tecnol\'ogica Federal do Paran\'a - Rua Cristo rei 19, Vila Becker, 85902-490, Toledo, PR, Brazil}

\author{Paola C. Obando}
\affiliation{Instituto de F\'isica, Universidade Federal Fluminense, Av. Gal. Milton Tavares de Souza s/n, Gragoat\'a, 24210-346, Niter\'oi, RJ, Brazil}

\author{Marcelo S. Sarandy}
\affiliation{Instituto de F\'isica, Universidade Federal Fluminense, Av. Gal. Milton Tavares de Souza s/n, Gragoat\'a, 24210-346, Niter\'oi, RJ, Brazil}
\affiliation{Center for Quantum Information Science \& Technology and Ming Hsieh Department of Electrical Engineering,
University of Southern California, Los Angeles, California 90089, USA}

\date{\today}

\begin{abstract}
We provide a characterization of memory effects in non-Markovian system-bath 
interactions from a quantum information perspective. More specifically, we establish sufficient  
conditions for which generalized measures of multipartite quantum, classical, and total correlations 
can be used to quantify the degree of non-Markovianity of a local quantum decohering process.  We 
illustrate our results by considering the dynamical behavior of the trace-distance correlations in 
multi-qubit systems under local dephasing and generalized amplitude damping.
\end{abstract}

\pacs{03.67.-a, 03.65.Yz, 03.65.Ud}

\maketitle


\section{Introduction}

The characterization of memory effects in non-Markovian evolutions is a 
fundamental subject in the theory of open systems~\cite{Breuer:book}. Indeed, the Markovian 
behavior is always an idealization in the description of the quantum dynamics, with non-Markovianity 
being non-negligible in a number of different scenarios, such as 
biological~\cite{Ishizaki:09,Rebentrost:09,Liang:10,Chen:15} or condensed matter 
systems~\cite{Wolf:08,Apollaro:11,Haikka:12}. From an applied point of view, non-Markovian 
dynamics may be a resource for quantum tasks through an increase in the capacities
of quantum channels~\cite{Bylicka:13}. Moreover, it also exhibits applications in fault-tolerant 
quantum computation~\cite{Aharonov:06}. Rigorously, non-Markovianity can be defined 
through the deviation of a  dynamical evolution map from a {\it divisible} completely positive 
trace-preserving (CPTP) map~\cite{Breuer:09}. This behavior is manifested both in 
entanglement~\cite{Rivas:10} and in other correlation sources~\cite{Luo:12,Haikka:13,Fanchini:14,He:14,Dhar:15},   
providing an approach that takes advantage of quantum information tools in the open-systems realm. 
More specifically, non-Markovianity can be interpreted in this context as a flow of information back to the system 
due to its interaction with the environment, which may imply non-monotonic behavior of 
correlations as a function of time. 

Here, we propose a general framework to characterize non-Markovianity through multipartite measures of quantum, 
classical, and total correlations. Recently, there has been growing interest in the investigation of correlation measures 
from a quantum information perspective~\cite{Modi,Celeri,Sarandy:12}. In particular, measures for quantum correlations 
such as those provided by discordlike quantities and their corresponding classical counterparts have been formulated, 
with these quantities used as resources for implementing a variety of quantum tasks (see Ref.~\cite{Modi} and references therein). 
These measures were originally introduced in an entropic scenario by Ollivier and Zurek~\cite{Ollivier:01}. 
In a geometric context, correlations can be defined based on the relative entropy~\cite{Modi:10}, Hilbert-Schmidt 
norm~\cite{Dakic:10,Bellomo:12}, trace norm~\cite{Paula,Nakano:12}, and Bures norm~\cite{Spehner:13,Bromley}. 
All of these distinct versions of correlation measures can be described by a unified framework in terms of a generalized distance 
(or pseudo distance) function~\cite{Modi,Brodutch-Modi,EPL108}. Our approach for non-Markovianity
includes all these measures as particular cases and characterizes the non-Markovian behavior by taking into account 
the program introduced by Rivas, Huelga, and Plenio~\cite{Rivas:10}, in which an ancilla is coupled to a system 
which interacts with an environment. Specifically, we provide a rigorous description of the hypotheses over local dynamical 
 maps under which quantum, classical, and total correlations can be used as measures of non-Markovianity. 
We illustrate the results considering the dynamics of the trace-distance correlations in qubit systems under either 
local dephasing or generalized amplitude damping (GAD). 


\section{Characterizing non-Markovianity}

Let us suppose a quantum process governed by a time-local master equation 
\begin{equation}
\displaystyle\frac{d \rho}{dt} = {\cal L}_t \, \rho(t), 
\label{master}
\end{equation}
where the time-dependent generator ${\cal L}_t$ is given by
\begin{eqnarray}
 {\cal L}_t \rho(t) &=& -i\left[H(t),\rho(t)\right]  + \sum_i \gamma_i (t) \left( A_i(t)\rho(t)A^\dagger_i(t) \right. \nonumber \\
 && \left. - \frac{1}{2} \left\{A^\dagger_i(t) A_i(t),\rho(t)\right\}\right) ,
 \end{eqnarray}
with $H(t)$ denoting the effective system Hamiltonian, $A_i(t)$ denoting the Lindblad operators, and $\gamma_i(t)$ denoting  
the relaxation rates. By taking the relaxation rates as positive functions, i.e., $\gamma_i(t) \ge 0$, the generator 
${\cal L}_t$ assumes the Lindblad form~\cite{Lindblad} for each fixed $t \ge 0$. The master equation describes 
the dynamics of the density operator through $\rho(t)=\Phi_{t,\,\tau}\rho(\tau)$, with the CPTP map  
$\Phi_{t,\,\tau}$ given by $\Phi_{t,\,\tau} = {\cal T} \exp \left(\int_{\tau}^{t} dt^\prime{\cal L}_{t^\prime} \right)$, with 
${\cal T} $ denoting the chronological time-ordering operator. The dynamical map $\Phi_{t,\,\tau}$ then satisfies the 
divisibility condition $\Phi_{t,\,\tau}=\Phi_{t,\,r}\Phi_{r,\,\tau}$ $(t \ge r \ge \tau \ge 0$), which characterizes the Markovianity of the 
quantum process.  On the other hand, for  $\gamma_i(t) < 0$, the corresponding dynamical 
map $\Phi_{t,\,\tau}$ may not be CPTP for intermediate time intervals and the divisibility property of the overall CPTP dynamics is violated, 
which characterizes a non-Markovian behavior~\cite{Breuer:09,Rivas:10,Dariusz:14}.

If a function $F=F(\rho)$ is monotonically nonincreasing under divisible maps acting on $\rho$, i.e., 
$F(\Phi_{t,\,\tau}\rho(\tau))\leq F(\rho(\tau))$ when $\Phi_{t,\,\tau}=\Phi_{t,\,r}\Phi_{r,\,\tau}$, 
then $F(\rho)$ is monotonically nonincreasing with increasing time, namely, $dF(t)/dt\leq 0$. However, this is not always 
true for a non-Markovian process, in which the divisibility of the map is violated, with $dF(t)/dt>0$ being a 
straightforward non-Markovianity witness~\cite{Breuer:09}. Therefore, $F(\rho)$ can be employed to point out the breakdown of 
Markovianity and the degree of non-Markovianity can be defined by
\begin{equation}
N_{F}(\Phi)=\text{max}_{\rho(0)}\int_{\frac{d}{dt}F(t)>0}{\frac{d}{dt}F(t)dt} ,
\label{NF}
\end{equation}
with the maximization performed over all sets of possible initial states $\rho \left(0\right)$ and the integration 
extended over all time intervals for which $dF(\rho)/dt>0$. Numerically, we can write
\begin{equation}
N_{F}(\Phi)=\textrm{max}_{\rho(0)}\sum_{i}\left[\left\{F(\tau_{i+1})-F(\tau_i)\right\}\right],
\label{NF2}
\end{equation}
where $\{(\tau_i,\tau_{i+1})\}$ represents the set of all time intervals for which $F(t+\Delta t)>F(t)$. The maximization 
over $\rho(0)$ is not a trivial task. Nevertheless, it is always possible to find out lower bounds to $N_F(\Phi)$ by optimizing 
over any class of initial states, which leads to a qualitative assessment of the non-Markovianity of the map $\Phi$~\cite{Dhar:15}.


\section{Non-Markovianity through Correlation measures}

In the general approach introduced in 
Refs.~\cite{Modi,Brodutch-Modi,EPL108}, discordlike measures of quantum, classical, and total correlations of an 
$n$-partite system in a state $\rho$ are defined by the respective expressions
\begin{eqnarray}
Q(\rho)&=&K\left[\rho, \, M^{-}\rho\right], \label{mQ}\\
C(\rho)&=&K\left[M^{+}\rho,\, M^{+}\pi_{\rho}\right], \label{mC} \\
T(\rho)&=&K\left[\rho,\, \pi_{\rho}\right],\label{mT}
\end{eqnarray}  
where $K\left[\rho,\sigma\right]$ denotes a real and positive function that vanishes for $\rho=\sigma$, 
the operator 
$\pi_{\rho} =\textrm{tr}_{\bar{1}}\rho\otimes \text{tr}_{\bar{2}}\rho...\otimes \text{tr}_{\bar{n}}\rho$ 
represents the product of the local marginals  of $\rho$, and $M^{-}\rho$ and $M^{+}\rho$ are {\it classical} 
states obtained through measurement maps $M^{-}$ and $M^{+}$ that minimize $Q$ and 
maximize $C$, respectively. In particular, they will be taken here as local $n$-partite maps $M^{\pm}=M^{\pm}_1\otimes M^{\pm}_2\cdots \otimes M^{\pm}_n$, where $M^{\pm}_i \ne \mathbb{I}$ or $M^{\pm}_i = \mathbb{I}$ depending on whether the \textit{i}th partition is measured or unmeasured. 
Furthermore, we will define the measurements $\{M^{\pm}_i \ne \mathbb{I}\}$ as optimized complete sets of 
local {\it orthogonal projectors}. Since orthogonal projective measurements are adopted, we have $M^{\pm}_i M^{\pm}_i \rho = M^{\pm}_i \rho$.

The correlation measures in Eqs.~(\ref{mQ})-(\ref{mT}) are expected to obey the following set of fundamental 
criteria~\cite{Modi,Brodutch-Modi,EPL108}: (i) product states have no correlations, (ii) all correlations are invariant 
under local unitary operations, (iii) all correlations are non-negative, (iv) total correlations are nonincreasing under 
local quantum channels (CPTP maps), (v) classical states have no quantum correlations, and (vi) quantum correlations 
are nonincreasing under local quantum channels over unmeasured subsystems~\cite{EPL108}. 
In order to satisfy the requirements above, we restrict $K$ to be 
positive and unitary invariant.  
Moreover, we also require $K$ to be {\it contractible} under CPTP maps $\Phi$~\cite{EPL108}, i.e., 
\begin{equation}
K[\Phi\rho, \Phi\sigma] \le K[\rho, \sigma] \,\,\,\,\,\, (\forall \rho,\sigma). 
\end{equation}
In particular, a variety of such functions $K$ can be adopted as, 
for instance, the trace distance $K (\rho, \sigma) = \mathrm{tr}|\rho-\sigma|$.

In this work, we will consider multipartite correlated quantum systems, such as that illustrated in Fig.~\ref{fig:1}. In particular, we 
will consider local dynamical maps $\Phi=\Phi_1\otimes\Phi_2 \cdots \otimes \Phi_n$. In this scenario, we can show that a 
non-monotonic behavior of the correlations $Q(\rho)$, $C(\rho)$, and $T(\rho)$ as a function of time may provide a direct
 measure of the degree of non-Markovianity $N_F(\Phi)$, with $F = Q, C,$ or $T$.  
This result is contained in Theorem~1 below. 

\begin{figure}[ht!]
\centering
\includegraphics[scale=0.33]{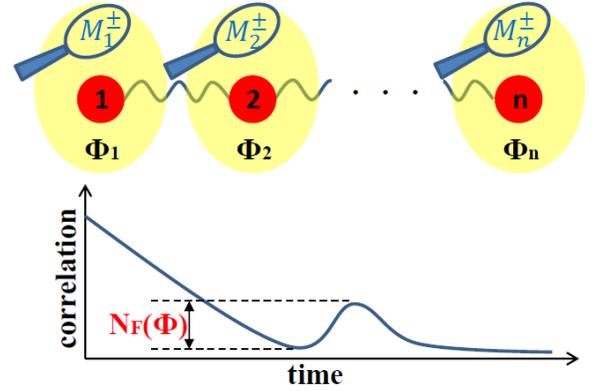}
\caption{Illustration of the degree of non-Markovianity $N_F(\Phi)$ via locally measured states.}
\label{fig:1}
\end{figure}

\begin{theorem}
Consider a quantum evolution driven by a local dynamical map $\Phi=\bigotimes_{i=1}^{n}\Phi_i$. 
Then, assuming that $K$ is contractible under CPTP maps, it follows that (i) $N_T(\Phi)$ is a measure 
of non-Markovianity and (ii) $N_Q(\Phi)$ and $N_C(\Phi)$ are measures of non-Markovianity for local 
measurements $M^{\pm}=\bigotimes_{i=1}^{n}M^{\pm}_i$ such that $M^{\pm}_i =\mathbb{I}$ when $\Phi_i \neq\mathbb{I}\,$.
\end{theorem}
\begin{proof}
We begin by using $\Phi=\bigotimes_{i=1}^{n}\Phi_i$. Then, by adopting $t \ge \tau \ge 0$, 
we write $\rho(t)=\Phi_{t,\tau}\rho(\tau)=\bigotimes_i\Phi_{i\,t,\tau\,}\rho(\tau)$.  Moreover, for local CPTP maps, we have the relation $\pi_{\Phi_{t,\tau}\rho(\tau)}=\Phi_{t,\tau}\pi_{\rho(\tau)}$. 
(i) Thus, let us consider $N_T(\Phi)$. A generalized total correlation measure can be written as $T(t) = K\left[\rho(t),\, \pi_{\rho(t)}\right]= K\left[\Phi_{t,\tau}\rho(\tau),\, \Phi_{t,\tau}\pi_{\rho(\tau)}\right]$. Imposing the condition that $K$ is contractible under CPTP maps, we have $K\left[\Phi_{t,\tau}\rho(\tau),\, \Phi_{t,\tau}\pi_{\rho(\tau)}\right] \leq K\left[\rho(\tau),\, \pi_{\rho(\tau)}\right]= T(\tau)$. Hence, $T(t)\leq T(\tau)$.
(ii) Let us now take $N_Q(\Phi)$ and $N_C(\Phi)$. A generalized quantum correlation measure can be written as $Q(t)= K\left[\rho(t),M_{t}^{-}\rho(t)\right]$, where we are considering $M_{t}^{-}=\bigotimes_i M^{-}_{i\,t}$ with $M^{-}_i =\mathbb{I}$ when $\Phi_i \neq\mathbb{I}$ (consequently, $M^{-}\Phi=\Phi M^{-}$). As $M_{\tau}^{-}$ does not necessarily minimize $Q(t)$, we can write $Q(t)\leq K\left[\rho(t),M_{\tau}^{-}\rho(t)\right]=K\left[\Phi_{t,\tau}\rho(\tau),\Phi_{t,\tau}M_{\tau}^{-}\rho(\tau)\right]$. By taking $\Phi_{i\, t,\tau}$ as a CPTP map and imposing the condition that $K$ is contractible under CPTP maps, we obtain $K\left[\Phi_{t,\tau}\rho(\tau),\Phi_{t,\tau}M_{\tau}^{-}\rho(\tau)\right]\leq K\left[\rho(\tau),M_{\tau}^{-}\rho(\tau)\right]=Q(\tau)$. Hence, $Q(t)\leq Q(\tau)$. Next, let us consider $N_C(\Phi)$. A generalized classical correlation measure can be written as $C(t)= K\left[M_{t}^{+}\rho(t),M_{t}^{+}\pi_{\rho(t)}\right]$, where we are considering $M_{t}^{+}=\bigotimes_i M^{+}_{i\,t}$ with $M^{+}_i =\mathbb{I}$ when $\Phi_i \neq\mathbb{I}$ (consequently, $M^{+}\Phi=\Phi M^{+}$). Then, imposing the condition that $K$ is contractible under CPTP maps, we can write $C(t)= K\left[\Phi_{t\, ,\tau}M_{t}^{+}\rho(\tau),\Phi_{t\, ,\tau}M_{t}^{+}\pi_{\rho(\tau)}\right]\leq K\left[M_{t}^{+}\rho(\tau), M_{t}^{+}\pi_{\rho(\tau)}\right]$. As $M_{t}^{+}$ does not necessarily maximize $C(\tau)$, then $K\left[M_{t}^{+}\rho(\tau), M_{t}^{+}\pi_{\rho(\tau)}\right]\leq K\left[M_{\tau}^{+}\rho(\tau),M_{\tau}^{+}\pi_{\rho(\tau)}\right]=C(\tau)$. Hence, $C(t)\leq C(\tau)$. 
\end{proof}
We observe that Theorem~1 ensures $N_Q(\Phi)$ and $N_C(\Phi)$ as measures of 
non-Markovianity by assuming that the subsystems under decoherence are unmeasured, i.e.,  
$M^{\pm}_i =\mathbb{I}$ when $\Phi_i \neq\mathbb{I}$. The measurements are then 
 performed over ancillary states that are effectively free of decoherence. 
This requirement is unnecessary for $N_T(\Phi)$ since it is a measurement-independent quantifier. 
An effective isolation of an ancilla to probe non-Markovianity has been experimentally achieved in 
several scenarios (see, e.g., Refs.~\cite{Xu:13,Bernardes:15}). More generally, it can be 
approximately assumed when the relaxation times of the ancillary subsystem are much larger 
than those of the principal subsystem. Similarly, it also happens when the multi-local dynamical map can 
be written as an effective transformation where only part of the subsystems undergoes 
decoherence [see Eq. (\ref{phi-eff-IVA}) for a bipartite example in Sec. IV B].


\section{Applications}


\subsection{Trace-norm correlations for two-qubit \textit{X} states}

Let us consider correlations based on the 
Schatten 1-norm (trace-norm) and projective measurements operating over one qubit within a two-qubit 
system, i.e., $K\left[\rho,\sigma\right]=\mathrm{tr}|\rho-\sigma|$ and $M^{\pm}=M_1^{\pm}\otimes\mathbb{I}$. 
By adopting these conditions, $Q(\rho)=\mathrm{tr}\left|\rho-M_1^{-}\otimes\mathbb{I}\,\rho\right|$ 
is then the trace-norm geometric quantum discord, as introduced in Refs.~\cite{Paula,Nakano:12}, 
with $C(\rho)=\mathrm{tr}\left|M_1^{+}\otimes \mathbb{I}\,\left(\rho-\pi_{\rho}\right)\right|$ and 
$T(\rho)=\mathrm{tr}\left|\rho-\pi_{\rho}\right|$ being the corresponding classical and total correlations~\cite{EPL108}. 
The trace-norm geometric quantum and classical correlations have been experimentally discussed in 
Refs.~\cite{Silva:13,PRL:13}. We will consider initial density operators restricted to the \textit{X} states, i.e.,
\begin{equation}
\rho(0)=\frac{1}{4}\left[\mathbb{I}\otimes\mathbb{I}+\sum_{i=1}^{3}c_i\,\sigma_i\otimes\sigma_i + 
c_4\mathbb{I}\otimes\sigma_3+c_5\sigma_3\otimes\mathbb{I}\right],
\end{equation}
where $\{\sigma_i\}$ represent the Pauli matrices and $\{c_{i}=c_{i}(0)\}$ are the initial correlation parameters. 
These states represent, for example, the general form of reduced density operators of arbitrary quantum spin 
chains with $Z_2$ (parity) symmetry  (for a review see, e.g., Ref.~\cite{Sarandy:12}). The trace-norm correlation 
measures  have been analytically developed for a general \textit{X} state~\cite{Ciccarello:14,Paola:15}, reading
\begin{eqnarray}
Q&=&\sqrt{\frac{ac-bd}{a+c-b-d}},
\label{quantum} \nonumber \\
C&=&\text{max}\left\{|c_1|,|c_2|,|c_3-c_4c_5|\right\}, \nonumber \\
T&=&\max\left\{C,\tfrac{1}{2}\left(|c_1|+|c_2|+|c_3-c_4c_5|\right)\right\},
\label{total}
\end{eqnarray}
where $a=\text{max}\{c_{3}^{2},d+c_{5}^{2}\}$, $b=\text{min}\{c,c_{3}^{2}\}$, $c=\text{max}\{c_{1}^{2},c_{2}^{2}\}$, 
and $d=\text{min}\{c_{1}^{2},c_{2}^{2}\}$.


\subsection{Two-qubit state under dephasing via a time-local master equation}

Let us start by focusing on a dynamical map of the form $\Phi=\Phi_{1}\otimes\Phi_{2}$, where
\begin{equation}
\Phi_j \rho_j = {\frac{1+f_j(t)}{2}} \rho_j \,\, + \,\, {\frac{1-f_j(t)}{2}}\sigma_3 \rho_j \sigma_3,
\label{local-dep}
\end{equation}
with $f_j(t)=\text{exp}\left[-2\int_{0}^{t}\gamma_j(\tau)d\tau\right]$. This represents a local dephasing channel acting 
over subsystem $j$, which can be derived from Eq.~(\ref{master}) by taking $H=-\left(\sigma_z^1+\sigma_z^2\right)/2$ and $A_j = \sigma_z^j$ ($j=1,2$). 
For simplicity, the time-dependent decoherence 
rates associated with each subsystem will 
be chosen to be the same, namely, $\gamma_1(t)=\gamma_2(t)$ such that $f_1(t)=f_2(t)$.
The map $\Phi_j$ preserves the \textit{X} state form of $\rho$, with  
$\{c_1(t),c_2(t),c_3(t),c_4(t),c_5(t)\}=\{c_1f_k(t),c_2f_k(t),c_3,c_4,c_5\}$. Moreover, we can write
\begin{equation}  \label{phi-eff-IVA}
\Phi=\Phi_{1}\otimes\Phi_{2}=\Phi^{eff}\otimes\mathbb{I}=\mathbb{I}\otimes\Phi^{eff},
\end{equation}
where $\Phi^{eff}$ is an effective dephasing channel with $f(t)=f_1(t)f_2(t)=\text{exp}\left[-2\int_{0}^{t}\gamma(\tau)d\tau\right]$, 
with $\gamma(\tau)=\gamma_1(\tau)+\gamma_2(\tau)=2\gamma_1(t)$, such that
$\{c_1(t),c_2(t),c_3(t),c_4(t),c_5(t)\}=\{c_1f(t),c_2f(t),c_3,c_4,c_5\}$. Therefore, $Q(t)$, $C(t)$, or $T(t)$ can be used to characterize non-Markovianity, as provided by Theorem~1.
In the Markovian regime, which takes place when $\gamma(t)\geq0$ for all $t\geq 0$, we have $df(t)/dt \leq 0$. 
Consequently, we have $dQ(t)/dt \leq 0$, $dC(t)/dt \leq 0$, 
and $dT(t)/dt \leq 0$. 

In order to quantify the non-Markovianity of the dynamical map $\Phi=\Phi_1\otimes\Phi_2$, we will first consider the 
classical correlation. Thus, we replace $F(t)$ by $C(t)$ in Eq.~(\ref{NF}).
The condition $dC(t)/dt>0$ occurs only if $C(t)>|c_3(0)-c_4(0)c_5(0)|$ and $\gamma(t)<0$. Under this condition, 
$dC(t)/dt=-2\text{max}\{|c_1(0)|,|c_2(0)|\}\gamma(t)f(t)$ and the maximization in Eq.~(\ref{NF}) is achieved when 
$c_3(0)-c_4(0)c_5(0)=0$ and $\text{max}\{|c_1(0)|,|c_2(0)|\}=1$. Thus, we find
\begin{equation}
N_{C}(\Phi)=-2\int_{\gamma(t)<0}{\gamma(t)f(t)dt} .
\label{NC}
\end{equation}
By taking $\rho(0)$ as a maximally entangled state, we obtain the same estimation of the non-Markovianity degree via total or quantum correlation. 
In fact, we have for the four Bell states $\left|c_1\right|=\left|c_2\right|=\left|c_3\right|=1$ and $c_4=c_5=0$ 
such that $T(t)=\max\{1,f(t)+\frac{1}{2}\}$ and $Q(t)=f(t)$. As $df(t)/dt=-2\gamma(t)f(t)$ and $f(t)>0$, we conclude that 
$dT(t)/dt>0$ or $dQ(t)/dt>0$ is equivalent to $\gamma(t)<0$. Under this condition, $f(t)>1$ and $dT(t)/dt=dQ(t)/dt=-2\gamma(t)f(t)$, 
which leads to
\begin{equation}
N_T(\Phi)=N_Q(\Phi)=N_{C}(\Phi).
\end{equation}
The degree of non-Markovianity found here via trace-distance correlation measures is then consistent with previous results for 
measures of non-Markovianity~\cite{Breuer:09, Dhar:15,Vacchini, Luo:12}.


\subsection{Two-qubit state under GAD via the Lindblad rate equation}

Let us now illustrate the signature of non-Markovianity for two qubits under local GAD, which will be described in 
the open-system framework developed in Ref.~\cite{AAPRA74} based on a Lindblad rate equation. In this scenario, the density matrix 
 in Eq.(\ref{master}) is replaced by 
\begin{equation}
\rho(t) = \sum^{R_{max}}_{R=1}\rho_R(t), 
\end{equation}
where each auxiliary (unnormalized) operator $\rho_R$ 
defines the system dynamics given that the reservoir is in the $R$-configurational bath state, with $R_{max}$ being the number of 
configurational states of the environment. The probability that the environment is in a given state at time $t$ reads 
$P_R(t)=\mathrm{tr}[\rho_R(t)]$, and the set of states $\{\rho_R(t)\}$ encodes both the system dynamics and the fluctuations of the 
environment~\cite{AAPRA74,HPPRA75}. Then, we model the environment as 
being characterized by a two-dimensional configurational space ($R_{max}=2$), which only affects the decay rates of the system.  
Each state follows by itself a Lindblad rate equation
\begin{eqnarray}\label{Eq:rho1}
\nonumber \frac{d\rho_1(t)}{dt}&=& -i[H_1, \rho_1(t)] + \bar{\gamma}^{A}_1\mathcal {L}^{A}\rho_1(t) + \bar{\gamma}^{B}_1\mathcal {L}^{B}\rho_1(t) \\&-& \phi_{21}\rho_1(t)+\phi_{12}\rho_{2}(t), \\
\label{Eq:rho2}
\nonumber \frac{d\rho_2(t)}{dt} &=& -i[H_2, \rho_2(t)] + \bar{\gamma}^{A}_2 \mathcal {L}^{A}\rho_2(t) + \bar{\gamma}^{B}_2 \mathcal {L}^{B}\rho_2(t) \nonumber\\ &-& \phi_{21}\rho_2(t)+\phi_{12}\rho_{1}(t),
\end{eqnarray}
where the structure of the superoperator $\mathcal {L}$ for the  GAD channel is given by
\begin{equation}
\mathcal {L}^{A,B} \rho_R =\left(-\frac{\sigma^{\dagger A,B} \sigma^{A,B}\,\rho_R}{2}- \frac{\rho_R\,\sigma^{\dagger A,B} \sigma^{A,B}}{2}+\sigma^{\dagger A,B}\rho_R \sigma^{A,B}\right).
\end{equation}
 The first lines of Eqs.~(\ref{Eq:rho1}) and (\ref{Eq:rho2}) define  the unitary and dissipative dynamics for the two-qubit system, given that the bath is in configurational state 
1 and configurational state 2, respectively.  The constants $\{\bar{\gamma}^{A}_{1,2},\bar{\gamma}^{B}_{1,2}\}$ are the natural decay rates of the system associated with each reservoir state. The positivity of the density matrix will be
ensured as long as these decoherence coefficients obey $\bar{\gamma}^{A\,B}_{1,2}\geq 0$ \cite{AAPRA74,budini43}.
On the other hand, the second line of Eqs. (\ref{Eq:rho1}) and (\ref{Eq:rho2}) describes transitions between the configurational states  of the environment (with rates $\phi_{12}$ 
and $\phi_{21}$)~\cite{budini43}. For simplicity, the decay rates associated with each subsystem will be chosen to be the same, namely, $\bar{\gamma}^{A}_{1}= \bar{\gamma}^{B}_{1}\equiv \bar{\gamma}_{1}$ and $\bar{\gamma}^{A}_{2}= \bar{\gamma}^{B}_{2}\equiv \bar{\gamma}_{2}$. Moreover, we define the characteristic dimensionless parameters  
\begin{equation}
\epsilon = \frac{\bar{\gamma}_1}{\bar{\gamma}_1 + \bar{\gamma}_2},\quad 
\eta = \frac{\phi_{12}}{\phi_{12}+ \phi_{21}} \quad 
v = \frac{ \phi_{12}+ \phi_{21}}{\bar{\gamma}_1 + \bar{\gamma}_2} ,
\end{equation}
where $\epsilon, \eta \in [0,1]$ and $v \in [0,\infty)$.
We will analyze the system in the limit of either fast or slow environmental fluctuations. The fast limit of environmental fluctuations occurs when the reservoir fluctuations are much faster than the average decay rates of the system, namely, $\{ \phi_{R'R}\} \gg  \{ \bar{\gamma}_R \}$ ($v\gg 1$), which implies that the system exhibits Markovian behavior. On the other hand, when the bath fluctuations are much slower than the average decay rate, namely, $\{ \phi_{R'R}\} \ll  \{\bar{\gamma}_R \}$ ($v \ll 1$), the system is in the limit of slow environmental fluctuations. The signatures of non-Markovianity will be provided by the total correlation, which has the advantages of avoiding both extremization procedures and further requirements 
over the dynamical map. The non-Markovian behavior can then be witnessed in  
Fig.~\ref{fig:GADvel}, which shows the temporal evolution of the total correlation for several values of  $v$, where we have taken $\epsilon=0.92$, $\eta=0,5$, and an
initial \textit{X} state described by $c_1=0.20, c_2=-0.20, c_3=0.60, c_4=0.50$, and $c_5=0.70$. 
\begin{figure}[h]
 \centering
 \includegraphics[scale=0.39]{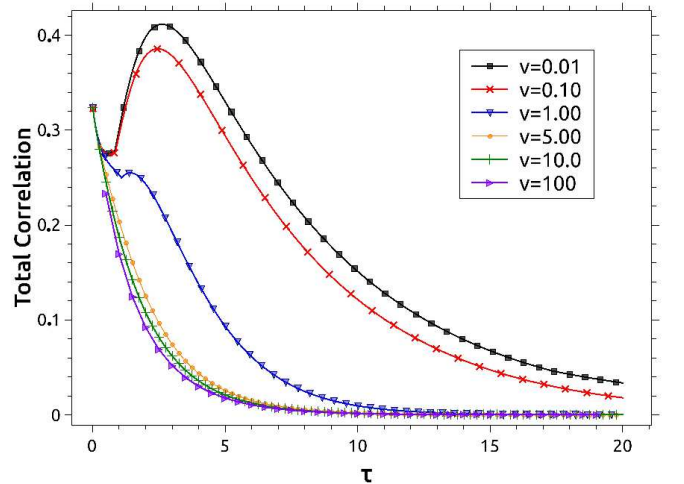} 
 \caption{Total correlation as a function of  $\tau=(\bar{\gamma}_A+\bar{\gamma}_B)t$ for a two-qubit system under a non-Markovian GAD channel. The initial state is in the \textit{X} form, with $c_1=0.20$, $c_2=-0.20$, $c_3=0.60$, $c_4=0.50$, and $c_5=0.70$. We have also taken $\epsilon = 0.92$ and $\eta=0.5$. The non-monotonic behavior 
gets pronounced as we decrease $v$  from the Markovian regime ($v \gg 1$) towards the non-Markovian regime.}
\label{fig:GADvel}
\end{figure}

For the fast limit, the decay of the total correlation is a monotonically decreasing function, which corresponds to a Markovian evolution. Otherwise, when the system is subject to a non-Markovian evolution (slow limit), the total correlation shows a non-monotonic evolution, which gets more pronounced as we decrease $v$. The degree of non-Markovianity $N_T(v)$ can be rigorously obtained from Eq.(\ref{NF2}) by 
a maximization 
over all initial states. On the other hand, a lower bound for $N_T(v)$  can be directly obtained from Fig.~\ref{fig:GADvel}  through the height of the non-monotonic sector as a function of $v$.


\subsection{Multipartite entangled state under local dephasing}

Consider an $n$-partite system initially in the Greenberger-Horne-Zeilinger (GHZ) state, 
i.e., a maximally entangled state of the form 
\begin{equation}
\rho=\frac{1}{2}\left(\left|0\right\rangle^{n}\left\langle 0\right|^{n}+\left|0\right\rangle^{n}\left\langle 1\right|^{n}+\left|1\right\rangle^{n}\left\langle 0\right|^{n}+\left|1\right\rangle^{n}\left\langle 1\right|^{n}\right),
\end{equation} 
where $\left|k\right\rangle^{n}=\left|k\right\rangle_{1}\otimes\left|k\right\rangle_{2}\otimes\cdots\otimes\left|k\right\rangle_{n}$ $\left(k=0,\,1\right)$. 
By applying a dynamical map 
$\Phi=\Phi_1\otimes\Phi_2\otimes\cdots\otimes\Phi_n$ 
over the GHZ state and choosing $\Phi_i$ as a local 
dephasing channel as in Eq.~(\ref{local-dep}), we get 
\begin{equation}
\rho(t)=\frac{1}{2}\left(\left|0\right\rangle^{n}\left\langle 0\right|^{n}+f(t)\left|0\right\rangle^{n}\left\langle 1\right|^{n}+f(t)\left|1\right\rangle^{n}\left\langle 0\right|^{n}+\left|1\right\rangle^{n}\left\langle 1\right|^{n}\right),
\end{equation} 
where $f(t)=\text{exp}\left[-2\int_{0}^{t}\gamma(\tau)d\tau\right]$, with $\gamma(t)=\sum_{i=1}^{n}{\gamma_i(t)}$ denoting the sum of 
the time-dependent decoherence rates. We will consider the multipartite total correlation as the non-Markovianity quantifier and use 
the GHZ state to provide a lower bound for  $N_{T}(\Phi)$. The product
of the local marginals of $\rho(t)$ is given by $\pi_{\rho(t)}=\mathbb{I}/2^{n}$, and the eigenvalues of the operator $\rho(t)-\pi_{\rho(t)}$ are
\begin{eqnarray} 
\lambda_i &=& -2^{-n}\,\,\,\,\, \left(1\leq i\leq 2^n-2\right), \nonumber \\
\lambda_{2^{n}-1} &=& \frac{1}{2}\left(1-2^{1-n}-f\right),  \nonumber \\ 
\lambda_{2^{n}} &=& \frac{1}{2}\left(1-2^{1-n}+f\right).
\end{eqnarray}
Therefore, $T(t)=\sum_i^{2^{n}}\left|\lambda_i\right|$ or
$T(t)=1-2^{1-n}+\max\{1-2^{1-n},f(t)\}.$ 
We have $dT/dt=0$ for $f(t)\leq 1-2^{1-n}$ and $dT/dt=-2\gamma(t)f(t)$ for $f(t) > 1-2^{1-n}$. Since $f(t)>0$, we conclude that $dT/dt>0$ is 
equivalent to the condition $\gamma(t)<0$. Moreover, $f(t)>1$ [consequently, $f(t)>1-2^{n-1}$] when $\gamma(t)<0$. Thus, we find a generalization of 
the result in Eq.(\ref{NC}):
\begin{equation}
N_{T}(\Phi)=\int_{\frac{d}{dt}T(t)>0}{\frac{d}{dt}T(t)dt}=-2\int_{\gamma(t)<0}{\gamma(t)f(t)dt}.
\end{equation}


\section{Conclusion}

We have introduced a unified framework based on generalized quantum, classical, and total 
correlation measures to characterize the non-Markovianity of local dynamical maps over multipartite quantum systems. 
This approach establishes sufficient conditions under which each class of correlation can be used to determine the degree of non-Markovian 
behavior.  We illustrated our results for different master-equation methods and for different sources of decoherence. We expect 
applications in experimental setups for which correlations may be accessible to the observer. In addition, the vanishing of entanglement for 
high-temperature regimes~\cite{Werlang:10} or for distant neighbors within a composite system~\cite{Maziero:PLA} may also 
motivate the use of generalized correlations as a tool to characterize non-Markovianity. 
Further applications include the assessment of other 
approaches beyond Markovianity (see, e.g., Ref.~\cite{Shabani:05}) and of additional axioms over correlation functions (see, e.g., Refs.~\cite{Cianciaruso,Hu:12}). 
These topics are left for future research. 


\vspace{1.2cm}

\section*{Acknowledgments} 

M.S.S. thanks D. Lidar for his hospitality at the University of Southern California.  
This work is supported by the Brazilian agencies CNPq, CAPES, and FAPERJ and the Brazilian National Institute for 
Science and Technology of Quantum Information (INCT-IQ).

\end{document}